\documentclass[11pt]{article}
\usepackage[a4paper,margin=2.6cm]{geometry}
\usepackage{amsmath,amssymb,amsthm}
\usepackage{algorithm,algpseudocode}
\usepackage{hyperref}

\newtheorem{theorem}{Theorem}[section]
\newtheorem{lemma}[theorem]{Lemma}
\newtheorem{proposition}[theorem]{Proposition}
\newtheorem{corollary}[theorem]{Corollary}
\theoremstyle{definition}
\newtheorem{definition}[theorem]{Definition}
\newtheorem{example}[theorem]{Example}
\theoremstyle{remark}
\newtheorem{remark}[theorem]{Remark}

\newcommand{\Q}{\mathbb Q}
\newcommand{\Z}{\mathbb Z}
\newcommand{\A}{\mathbb A}
\newcommand{\Cl}{\operatorname{Cl}}
\newcommand{\cO}{\mathcal O}
\newcommand{\Fbar}{\overline{F}}
\newcommand{\den}{\operatorname{den}}
\newcommand{\Const}{\operatorname{Const}}
\newcommand{\Res}{\operatorname{Res}}

\newcommand{\Div}{\operatorname{Div}}
\newcommand{\Jac}{\operatorname{Jac}}
\newcommand{\disc}{\operatorname{disc}}
\newcommand{\rk}{\operatorname{rk}}
\newcommand{\dv}{\operatorname{div}}
\newcommand{\Dbar}{\overline{D}}
\newcommand{\fp}{\mathfrak p}
\newcommand{\fm}{\mathfrak m}

\title{Parallel Integration over Simple Radical Extensions}
\author{Sam Blake}
\date{\today}

\begin{document}
\maketitle

\begin{abstract}
The parallel Risch (Risch--Norman) method is a fast heuristic for computing elementary integrals over towers of transcendental extensions. Its justification rests on two structural facts about the integral: a bound on its denominator and a description of the logarithms that can occur. Both are known for purely logarithmic towers (Davenport--Trager) and, in the form of a structure theorem, for arbitrary derivations on multivariate rational function fields (Bronstein). We extend both facts to a simple radical extension $L=K(y)$, $y^m=q$, of such a field. The key observations are that the integral closure of $F[t_1,\dots,t_n]$ in $L$ has an explicit basis, so that all factorisation can remain in a polynomial ring, and that the derivation has a well-defined pole order $\delta_P\in\{0,1,e_P\}$ at every height-one prime $P$, so that pole orders of derivatives shift by $\delta_P$. The denominator of the integral then has the same Hermite-type shape as in the transcendental case, while the admissible logands are precisely the $S$-units of the integral closure for an explicit finite set $S$ of primes; the latter can be larger than the set generated by irreducible polynomials, as the unit $x+\sqrt{x^2+1}$ shows. For $n=1$ we relate these $S$-units to torsion in the Jacobian and, for $m=2$, to the polynomial Pell equation, obtaining a complete description of the logarithmic part in genus~0. We describe the resulting algorithm and give examples.
\end{abstract}

\section{Introduction}

The Risch--Norman method \cite{NormanMoore77,Davenport82,DavenportTrager85,GeddesStefanus89,Norman90,Bronstein05,Bronstein07} computes elementary antiderivatives of functions lying in a differential field $K=F(t_1,\dots,t_n)$ by treating all generators $t_i$ simultaneously: it guesses the denominator of the integral and the arguments of the logarithms that may occur, bounds the degree of the remaining polynomial numerator, and solves a linear system for its coefficients. It is fast and easy to implement, and is used either as a preprocessor for, or in place of, the complete Risch algorithm in several computer algebra systems. It is, however, only a heuristic: to turn it into an algorithm for a class of integrands one must prove that the guesses are exhaustive. Davenport and Trager \cite{DavenportTrager85} proved this for the denominator and the logarithms in purely logarithmic towers; Bronstein \cite[Ch.~10]{Bronstein05} gave a structure theorem for arbitrary derivations on $F(t_1,\dots,t_n)$ and sharpened it for a class of ``simple'' differential fields including nested logarithms.

All of this theory is confined to the case where $K$ is a purely transcendental extension of the constant field, so that $F[t_1,\dots,t_n]$ is a unique factorisation domain. Integration in algebraic extensions is handled by an entirely different family of algorithms going back to Risch \cite{Risch70}, Trager \cite{Trager84} and Bronstein \cite{Bronstein90}, based on integral bases, Hermite reduction and a torsion test in the Jacobian of a curve.

In this paper we show that, for a \emph{simple radical extension} $L=K(y)$ with $y^m=q$, $q\in F[t_1,\dots,t_n]$, the structural results underlying the parallel method survive essentially intact, and we identify precisely what changes. The two ingredients are:

\begin{enumerate}
\item The integral closure $\cO$ of $R=F[t_1,\dots,t_n]$ in $L$ is a free $R$-module with the explicit basis $y^i/\prod_j Q_j^{\lfloor ij/m\rfloor}$, where $q=\prod_j Q_j^{\,j}$ is the squarefree decomposition (Section~\ref{sec:basis}). This is Trager's basis for the univariate case \cite{Trager84}, used also in \cite{Bronstein98}; we give a short proof that it remains an integral basis in the multivariate setting. As a consequence every element of $L$ has a canonical denominator in $R$, and all gcd, squarefree and splitting factorisations required by the parallel method take place in the unique factorisation domain $R$.
\item At every height-one prime $P$ of $\cO$ lying over a normal prime of $R$, the derivation $D$ has a pole order $\delta_P\in\{1,e_P\}$, where $e_P$ is the ramification index, and $v_P(Dg)=v_P(g)-\delta_P$ whenever $v_P(g)\neq0$ (Section~\ref{sec:valuations}). This single lemma replaces both Davenport--Trager's Theorem~1 and Bronstein's Lemma~10.2.1, and it yields the same Hermite-type denominator bound $\den(\int f)\mid s\prod_j d_j^{\,j-1}$ as in the transcendental case (Section~\ref{sec:structure}).
\end{enumerate}

The description of the logarithms is where the algebraic case genuinely differs. In a unique factorisation domain the logands can be taken to be irreducible polynomials, and Davenport--Trager's Lemma~3 lists them. In $\cO$, which is generally not a UFD and may have non-constant units, the correct statement (Theorem~\ref{thm:sunits}) is that the logands are $S$-units of $\cO$, where $S$ is the set of primes at which the integrand has a pole of order at least $\delta_P$, together with the special primes. For example $\int dx/\sqrt{x^2+1}=\log\bigl(x+\sqrt{x^2+1}\bigr)$, and $x+\sqrt{x^2+1}$ is a unit of norm~1. For $n=1$ we describe $S$-units via an exact sequence with the Jacobian (Section~\ref{sec:curve}); for $m=2$ this is the polynomial Pell equation and, by a theorem of Abel, the periodicity of the continued fraction of $\sqrt q$. In genus $0$ the logarithmic part is thereby completely determined, and in general it reduces to Trager's torsion test. For $n=1$ we also prove exact degree bounds for the numerator of the rational part, so that the method becomes a complete algorithm over $\bar F(x,y)$; for towers we prove the degree bound in the top variable and show by example that the classical ``$+1$'' bound fails for lower variables, even in irreducibly logarithmic towers.

Section~\ref{sec:degree} treats Risch--Norman condition (ii): for $n=1$ the same valuation argument at the places at infinity gives exact, provable degree bounds (Corollary~\ref{cor:degbound}), which together with the results above make the method complete over $\bar F(x,y)$ (Theorem~\ref{thm:complete}); for towers we prove a transfer theorem in the top variable (Theorem~\ref{thm:top}) and give explicit nested-logarithmic examples in which the integral has degree two more than the integrand in a lower variable, violating condition (ii) of \cite{DavenportTrager85}. Section~\ref{sec:algorithm} turns these results into an algorithm which is Bronstein's \textsc{ParallelIntegrate} with the numerator ranging over $\cO$ and the logand list extended by $S$-units. One can view the outcome as ``Trager's algorithm with Hermite reduction replaced by undetermined coefficients''. Section~\ref{sec:examples} gives examples computed with a prototype implementation, and Section~\ref{sec:conclusion} lists open problems.

\section{Preliminaries}\label{sec:prelim}

All fields have characteristic $0$. We follow the notation of \cite{Bronstein05}. A differential field is a field $K$ with a derivation $D$; $\Const_D(K)=\{c\in K: Dc=0\}$. Throughout,
\[
K=F(t_1,\dots,t_n),\qquad DF\subseteq F,\qquad \Const_D(K)=F,
\]
with each $t_i$ transcendental over $F(t_1,\dots,t_{i-1})$, and $R=F[t_1,\dots,t_n]$, a UFD. The \emph{denominator of $D$}, $\den_D(K)\in R$, is the least common multiple of the denominators of $Dt_1,\dots,Dt_n$; then $\Dbar_R:=\den_D(K)\,D$ is a derivation of $R$ \cite[\S10.1]{Bronstein05}. An irreducible $p\in R$ is \emph{normal} if $\gcd(p,\Dbar_Rp)=1$ and \emph{special} if $p\mid\Dbar_Rp$; thus $p$ is normal iff $p\nmid\den_D(K)$ and $p\nmid Dp$, and every irreducible is one or the other. A \emph{splitting factorisation} $d=d_sd_n$ of $d\in R$ has $d_s$ special and every squarefree factor of $d_n$ normal; it is computed by Bronstein's \textsc{SplitFactor} \cite[\S10.1]{Bronstein05}. We write $\Fbar$ for the algebraic closure of $F$.

We use the strong Liouville theorem in the following form \cite[Thm.~5.5.3]{Bronstein05}: if $f$ in a differential field $L$ with $\Const_D(L)=F$ has an elementary integral over $L$, then there are $v\in L$, $c_1,\dots,c_r\in\Fbar$ and $u_1,\dots,u_r\in L(c_1,\dots,c_r)^*$ with
\begin{equation}\label{eq:liouville}
f = Dv+\sum_{i=1}^r c_i\frac{Du_i}{u_i}.
\end{equation}

\section{Simple radical extensions and their integral basis}\label{sec:basis}

\subsection{Normalisations}

Let $m\ge2$ and $q\in R\setminus F$, and let $L=K(y)$ with $y^m=q$. Write the squarefree decomposition of $q$ as
\begin{equation}\label{eq:sqfree}
q=\prod_{j\ge1}Q_j^{\,j},\qquad Q:=\prod_j Q_j,
\end{equation}
with the $Q_j\in R$ squarefree and pairwise coprime, and $Q$ the squarefree part of $q$. The decomposition \eqref{eq:sqfree} is unchanged under extension of the constant field. We impose two normalisations.

\begin{enumerate}
\item[(N1)] $q$ is $m$-th-power-free: $Q_j=1$ for $j\ge m$. (Otherwise replace $y$ by $y/\prod_{j\ge m}Q_j^{\lfloor j/m\rfloor}$.)
\item[(N2)] $y^m-q$ is irreducible over $\Fbar K$; equivalently, by Capelli's theorem \cite{Lang02}, $\gcd\bigl(m,\{\,j: Q_j\ne1\,\}\bigr)=1$.
\end{enumerate}
To see the equivalence in (N2): $q\in(\Fbar K)^p$ for a prime $p$ iff every exponent $j$ with $Q_j\neq1$ is divisible by $p$; and $-4\in(\Fbar K)^4$ since $-4=(1+i)^4$. If (N2) fails, $L$ either is a smaller radical extension or contains new constants, and one should replace $(q,m)$ by the radical actually generating $L$ over the algebraic closure of $F$ in $L$.

\begin{lemma}\label{lem:constants}
Under \textup{(N2)}, $F$ is algebraically closed in $L$, and $\Const_D(L)=F$ for the unique extension of $D$ to $L$, given by
\begin{equation}\label{eq:Dy}
\frac{Dy}{y}=\frac1m\frac{Dq}{q}=\frac1m\sum_j j\,\frac{DQ_j}{Q_j}.
\end{equation}
\end{lemma}
\begin{proof}
If $F\subsetneq F'\subseteq L$ with $[F':F]<\infty$, then $F'K\subseteq L$ and $[L:F'K]=m/[F':F]<m$, so $y$ has degree $<m$ over $F'K\subseteq\Fbar K$, contradicting (N2). Constants of an algebraic extension are algebraic over the constants of the base \cite[Cor.~3.3.1]{Bronstein05}, hence $\Const_D(L)\subseteq L\cap\Fbar=F$.
\end{proof}

Note that $y$ need not be transcendental over the intermediate fields $F(t_1,\dots,t_i)$: for $K=\Q(x,e^x)$ and $q=e^x$, $y=e^{x/2}$ is transcendental over $\Q(x)$ but algebraic over $K$. The theory must allow $y$ to behave both like a branch and like a special element.

\subsection{The integral basis}

\begin{proposition}\label{prop:basis}
Let $E_i:=\prod_j Q_j^{\lfloor ij/m\rfloor}$ and $w_i:=y^i/E_i$ for $0\le i<m$. Then
\[
\cO:=\bigoplus_{i=0}^{m-1}R\,w_i
\]
is the integral closure of $R$ in $L$. Moreover, for $0\le i,k<m$,
\begin{equation}\label{eq:multtable}
w_iw_k=w_{(i+k)\bmod m}\prod_j Q_j^{\,\lfloor (i+k)j/m\rfloor-\lfloor ij/m\rfloor-\lfloor kj/m\rfloor},
\end{equation}
with all exponents nonnegative.
\end{proposition}

\begin{proof}
Each $w_i$ is integral, since $w_i^m=\prod_jQ_j^{\,ij-m\lfloor ij/m\rfloor}\in R$. Formula \eqref{eq:multtable} is immediate for $i+k<m$; for $i+k\ge m$ use $y^{i+k}=q\,y^{i+k-m}$ and $j+\lfloor (i+k-m)j/m\rfloor=\lfloor (i+k)j/m\rfloor$. The exponents are nonnegative because $\lfloor a+b\rfloor\ge\lfloor a\rfloor+\lfloor b\rfloor$. Hence $\cO$ is a ring, and it is a free $R$-module, so it is Cohen--Macaulay and in particular satisfies Serre's condition $S_2$. By Serre's criterion it remains to show that $\cO$ is regular in codimension one, i.e.\ that $\cO_{(\fp)}$ is integrally closed for every height-one prime $\fp=(p)$ of $R$.

If $p\nmid q$, then $R_{(p)}[y]/(y^m-q)$ is \'etale over the discrete valuation ring $R_{(p)}$, its discriminant $\pm m^mq^{m-1}$ being a unit, so it is integrally closed and equals $\cO_{(p)}$.

If $p\mid Q_j$, put $g=\gcd(j,m)$ and $e=m/g$. We claim every prime $P$ of $L$ above $p$ has ramification index $e$. Let $\widehat K$ be the completion of $K$ at $p$ and write $q=p^ju'$ with $u'$ a unit of $R_{(p)}$. The element $z=y^e/p^{j/g}$ satisfies $z^g=u'$, so $\widehat K(z)/\widehat K$ is unramified (the residue characteristic is $0$); and $y^e=p^{j/g}z$ with $\gcd(j/g,e)=1$ forces $\widehat K(z)(y)/\widehat K(z)$ to be totally ramified of degree $e$. Since the primes of $L$ above $p$ correspond to the factors of $L\otimes_K\widehat K$, the claim follows. Ramification is tame, so by Dedekind's different theorem
\[
v_p\bigl(\disc(\cO_{(p)}/R_{(p)})\bigr)=\sum_{P\mid p}f_P(e_P-1)=\frac{e-1}{e}\,m=m-g .
\]
On the other hand, the discriminant of the basis $(w_i)$ is that of $(y^i)$ divided by $\prod_iE_i^2$, so its $p$-valuation is
\[
j(m-1)-2\sum_{i=0}^{m-1}\Bigl\lfloor\frac{ij}{m}\Bigr\rfloor
= j(m-1)-\bigl((m-1)(j-1)+g-1\bigr)=m-g,
\]
using the classical lattice-point count $\sum_{i=0}^{m-1}\lfloor ij/m\rfloor=\tfrac12\bigl((m-1)(j-1)+\gcd(m,j)-1\bigr)$. A basis of integral elements whose discriminant has the valuation of the discriminant of the extension is a local integral basis, so $\cO_{(p)}$ is the integral closure of $R_{(p)}$.
\end{proof}

We record the local information obtained in the proof.

\begin{corollary}\label{cor:local}
Let $P$ be a height-one prime of $\cO$ lying over the irreducible $p\in R$.
\begin{enumerate}
\item If $p\nmid Q$ then $e_P=1$ and $v_P(g)=v_p(g)$ for $g\in K$.
\item If $p\mid Q_j$ then $e_P=m/\gcd(j,m)$, $v_P(g)=e_Pv_p(g)$ for $g\in K$, $v_P(y)=je_P/m$, and $v_P(w_i)=e_P\{ij/m\}=(ij\bmod m)/\gcd(j,m)$, a nonnegative integer.
\end{enumerate}
\end{corollary}

Since $\cO$ is free over $R$, every $f\in L$ has a unique representation
\begin{equation}\label{eq:rep}
f=\frac{\sum_{i=0}^{m-1}a_iw_i}{d},\qquad a_i,d\in R,\quad \gcd(a_0,\dots,a_{m-1},d)=1,
\end{equation}
with $d$ normalised (say monic in a fixed ordering); we call $d=\den(f)$ the \emph{denominator} of $f$. Because $\cO$ is a Krull domain, $\cO=\bigcap_P\cO_P$ over its height-one primes, and $\cO_{(p)}=\bigcap_{P\mid p}\cO_P$. Hence, if $p^k\,\|\,\den(f)$, then $v_P(f)\ge -ke_P$ for all $P\mid p$, and $v_P(f)\le-(k-1)e_P-1$ for at least one $P\mid p$.

\subsection{The derivation on $\cO$}

From \eqref{eq:Dy},
\begin{equation}\label{eq:Dw}
\frac{Dw_i}{w_i}=\sum_j\Bigl\{\frac{ij}{m}\Bigr\}\frac{DQ_j}{Q_j},\qquad \{x\}:=x-\lfloor x\rfloor,
\end{equation}
so that $D\cO\subseteq H^{-1}\cO$ with
\begin{equation}\label{eq:H}
H:=\operatorname{lcm}\bigl(\den_D(K),\,Q\bigr).
\end{equation}
Thus the squarefree part of $q$ joins the denominator of the derivation, exactly as in the logarithmic case $Dt=Dq/q$ of \cite[p.~304]{Bronstein05}.

\section{Valuations and the derivation}\label{sec:valuations}

\begin{definition}
A height-one prime $P$ of $\cO$ is \emph{tame} if the irreducible $p\in R$ below it is normal (i.e.\ $p\nmid\den_D(K)$ and $p\nmid Dp$). For tame $P$ put
\[
\delta_P:=\begin{cases}1&\text{if }p\nmid Q,\\ e_P&\text{if }p\mid Q.\end{cases}
\]
\end{definition}

\begin{lemma}[pole order of $D$]\label{lem:poleorder}
Let $P$ be tame with uniformiser $\pi$. Then:
\begin{enumerate}
\item $v_P(D\pi/\pi)=-\delta_P$;
\item $v_P(Du)\ge 1-\delta_P$ for every unit $u$ of $\cO_P$;
\item for every $g\in L^*$, $v_P(Dg)\ge v_P(g)-\delta_P$, with equality if $v_P(g)\ne0$. More precisely, $Dg/g=v_P(g)\,D\pi/\pi+r$ with $v_P(r)\ge1-\delta_P$.
\end{enumerate}
\end{lemma}

\begin{proof}
Since $p\nmid\den_D(K)$, $D$ maps $R_{(p)}$ into itself, and since $p\nmid Dp$ we have $v_p(Dp)=0$, hence $v_P(Dp)=0$ and $v_P(Dp/p)=-e_P$.

\emph{Unramified case} ($p\nmid Q$, $\delta_P=1$). Take $\pi=p$; then (1) is $v_P(Dp/p)=-1$. For $a=\sum a_iw_i\in\cO$,
\begin{equation}\label{eq:Da}
Da=\sum_i(Da_i)w_i+\sum_i a_i\,\frac{Dw_i}{w_i}\,w_i,
\end{equation}
and both $Da_i\in R_{(p)}$ and $Dw_i/w_i$ (by \eqref{eq:Dw}, as $p\nmid Q$) are $P$-integral, so $D\cO_P\subseteq\cO_P$ and (2) holds with $1-\delta_P=0$.

\emph{Branch case} ($p\mid Q_j$, $\delta_P=e:=e_P=m/g$, $g=\gcd(j,m)$). Choose integers $\alpha,\beta$ with $\alpha(j/g)+\beta e=1$ and set $\pi=y^\alpha p^\beta$; by Corollary~\ref{cor:local}, $v_P(\pi)=\alpha je/m+\beta e=1$, so $\pi$ is a uniformiser. Writing $q=p^ju'$ with $p\nmid u'$,
\[
\frac{D\pi}{\pi}=\alpha\frac{Dy}{y}+\beta\frac{Dp}{p}=\Bigl(\frac{\alpha j}{m}+\beta\Bigr)\frac{Dp}{p}+\frac{\alpha}{m}\frac{Du'}{u'}=\frac1e\frac{Dp}{p}+\frac\alpha m\frac{Du'}{u'}.
\]
The first term has valuation $-e$ and the second is $P$-integral, which proves (1). For (2), consider \eqref{eq:Da}: the terms $(Da_i)w_i$ are $P$-integral, and in $a_iw_i\,Dw_i/w_i=a_iw_i\sum_l\{il/m\}DQ_l/Q_l$ only $l=j$ contributes a pole, of order $e$, with coefficient $a_iw_i\{ij/m\}$, which either vanishes or has valuation at least $v_P(w_i)=(ij\bmod m)/g\ge1$ by Corollary~\ref{cor:local}. Hence $v_P(Da)\ge1-e$ for all $a\in\cO$. A unit $u$ of $\cO_P$ is $a/b$ with $a,b\in\cO$, $v_P(a)=v_P(b)=0$, and $Du=Da/b-aDb/b^2$ has valuation $\ge1-e$.

Finally (3): write $g=u\pi^k$ with $u$ a unit and $k=v_P(g)$; then $Dg/g=k\,D\pi/\pi+Du/u$, and (1), (2) give the claim, the term $k\,D\pi/\pi$ having valuation exactly $-\delta_P$ when $k\ne0$ because the characteristic is $0$.
\end{proof}

\begin{remark}\label{rem:special}
(a) The argument uses nothing about the residue field of $P$; in particular it is unchanged for $n>1$.

(b) If $p\mid Q$ but $p$ is special in $R$ ($p\mid Dp$), for instance $q=e^x$, then $Dp/p$ and hence $D\pi/\pi$ are $P$-integral: $D\cO_P\subseteq\cO_P$ and $D\fm_P\subseteq\fm_P$. Such $P$ behave like the special (hyperexponential) primes of the transcendental theory: the pole order of a derivative does not grow, and cancellation is possible. This is the case $n_i'=n_i$ in Davenport--Trager's condition~(i).

(c) Primes over factors of $\den_D(K)$ are not tame; as in \cite[\S10.3]{Bronstein05} they are absorbed into a guessed special denominator.

(d) The three regimes are unified by describing $\delta_P$ as the pole order of the logarithmic derivative at $P$: $\delta_P=0$ at special primes, $1$ at normal unramified primes, $e_P$ at normal branch primes.
\end{remark}

\section{Structure of elementary integrals}\label{sec:structure}

Let $f\in L^*$ have an elementary integral over $L$, and fix a Liouville decomposition \eqref{eq:liouville}. Extending constants to $\Fbar$ does not affect the normalisations (N1), (N2) nor the squarefree decomposition \eqref{eq:sqfree}, so $\Fbar L=\Fbar K(y)$ is again a simple radical extension with the same integral basis over $\Fbar\otimes_FR$; its height-one primes lie over those of $\cO$, and Lemma~\ref{lem:poleorder} applies verbatim. All valuations below are taken in $\Fbar L$; the $\cO$-denominator of an element of $L$ is unchanged by this extension.

\paragraph{Rearranging the logarithmic part.}
Let $\gamma_1,\dots,\gamma_s$ be a $\Q$-basis of $\sum_i\Q c_i$, write $c_i=\sum_kr_{ik}\gamma_k$ with $r_{ik}\in\Q$, and let $N\in\Z_{>0}$ clear all denominators. Then
\begin{equation}\label{eq:rearr}
\sum_ic_i\frac{Du_i}{u_i}=\frac1N\sum_{k=1}^s\gamma_k\frac{DU_k}{U_k},\qquad U_k:=\prod_iu_i^{\,Nr_{ik}}\in(\Fbar L)^*,
\end{equation}
with $\gamma_1,\dots,\gamma_s$ linearly independent over $\Q$. This replaces the step in the transcendental theory where the $u_i$ are taken to be pairwise coprime irreducibles, which is unavailable in $\cO$.

\begin{theorem}[denominator]\label{thm:denominator}
Let $P$ be tame. Then
\[
v_P(v)=v_P(f)+\delta_P\quad\text{if } v_P(f)<-\delta_P,\qquad v_P(v)\ge0\quad\text{if } v_P(f)\ge-\delta_P.
\]
Moreover, if $v_P(f)>-\delta_P$ then $v_P(U_k)=0$ for all $k$.
\end{theorem}

\begin{proof}
Let $\Sigma=\sum_k\gamma_kDU_k/U_k$ and $\lambda_P\neq0$ the coefficient of $\pi^{-\delta_P}$ in $D\pi/\pi$. By Lemma~\ref{lem:poleorder}(3), $v_P(\Sigma)\ge-\delta_P$ and the coefficient of $\pi^{-\delta_P}$ in $\Sigma$ is $\lambda_P\sum_k\gamma_kv_P(U_k)$, which vanishes iff all $v_P(U_k)=0$, by $\Q$-linear independence of the $\gamma_k$.

If $v_P(v)=-k<0$, then $v_P(Dv)=-k-\delta_P<-\delta_P\le v_P(\Sigma)$ by Lemma~\ref{lem:poleorder}(3), so $v_P(f)=v_P(Dv)=-k-\delta_P$. If $v_P(v)\ge0$, then $v_P(Dv)\ge-\delta_P$ by Lemma~\ref{lem:poleorder}(1)--(3), so $v_P(f)\ge-\delta_P$. Together these give the first assertion. For the second, if $v_P(f)>-\delta_P$ then $v_P(v)\ge0$, so $Dv$ and $f$ both lack a $\pi^{-\delta_P}$ term (Lemma~\ref{lem:poleorder}(2),(3)); hence so does $\Sigma$, and all $v_P(U_k)=0$.
\end{proof}

\begin{corollary}[shape of the rational part]\label{cor:shape}
Let $d=\den(f)$, let $d=d_sd_n$ be its splitting factorisation in $R$, and let $d_n=\prod_jd_j^{\,j}$ be the squarefree factorisation of $d_n$. Then
\[
v=\frac{b}{s\prod_jd_j^{\,j-1}},\qquad b\in\cO,\ s\in R,
\]
where every irreducible factor of $s$ is special in $R$.
\end{corollary}

\begin{proof}
Every height-one prime of $\cO$ lies over a normal $p$ (then it is tame) or over a special $p$. For tame $P$ over $p$ with $p^k\|d$, we have $v_P(f)\ge-ke_P$, so Theorem~\ref{thm:denominator} gives $v_P(v)\ge\min(0,-ke_P+\delta_P)\ge-(k-1)e_P$, i.e.\ $v\in p^{-(k-1)}\cO_P$. Denominators at primes over special $p$ are unconstrained.
\end{proof}

Note that a branch prime $p\mid Q$ is normal in $R$ as soon as $p\nmid\den_D(K)$ and $p\nmid Dp$, so \textsc{SplitFactor} places it in $d_n$ and the corollary bounds it: the pole of the radical itself at $p$ has been absorbed by the integral basis. Corollary~\ref{cor:shape} is the exact analogue of \cite[Thm.~10.2.1]{Bronstein05}, and of \cite[Cor.~1]{DavenportTrager85} when there are no specials; the branch primes are ``simple with shift $e_P$'' in the sense of \cite[\S10.4]{Bronstein05}.

\begin{theorem}[logands are $S$-units]\label{thm:sunits}
Let
\[
S_f:=\{P\ \text{tame}: v_P(f)\le-\delta_P\}\ \cup\ \{P\ \text{not tame}\}.
\]
Then each $U_k$ in \eqref{eq:rearr} is an $S_f$-unit of $\Fbar\otimes\cO$: $v_P(U_k)=0$ for every height-one $P\notin S_f$. Consequently the logands in \eqref{eq:liouville} may be taken in $(\Fbar\otimes\cO)_{S_f}^\times$, and for tame $P$ with $v_P(f)\ge-\delta_P$ the integers $v_P(U_k)$ are determined by $f$ through the coefficient of $\pi^{-\delta_P}$ in $f$; for tame $P$ with $v_P(f)<-\delta_P$ they are determined by the same coefficient after subtracting a local derivative that reduces the pole order of $f$ at $P$ to at most $\delta_P$ (such a reduction exists at every tame prime, and the resulting coefficient does not depend on it; see \cite[Lemma~6.4]{PartII} --- for $n=1$ it is the ordinary residue of $f\,dx$, Section~\ref{sec:curve}).
\end{theorem}
\begin{proof}
This is the second assertion of Theorem~\ref{thm:denominator}. For the last claim, if $v_P(f)\ge-\delta_P$ then $v_P(v)\ge0$ by Theorem~\ref{thm:denominator}, so $v_P(Dv)\ge1-\delta_P$ by Lemma~\ref{lem:poleorder} and the $\pi^{-\delta_P}$ coefficient of $f$ is that of $\Sigma$, namely $\lambda_P\sum_k\gamma_kv_P(U_k)$. When $v_P(f)<-\delta_P$ the coefficient of $\pi^{-\delta_P}$ in $f$ alone is not meaningful (it receives contributions from $Dv$); subtracting a local derivative as in the statement removes them, and the invariance of the resulting coefficient is proved in \cite[Lemma~6.4]{PartII}.
\end{proof}

Theorem~\ref{thm:sunits} replaces \cite[Lemma~3, Thm.~3]{DavenportTrager85} and the corresponding part of \cite[Thm.~10.2.1]{Bronstein05}. When $\cO$ is a UFD with $\cO^\times=F^*$ --- the transcendental case --- $S$-units are generated by the primes in $S$, and one recovers the classical statement that the logands are the irreducible factors of $d_n$ and the specials. In general three new phenomena occur.

\begin{enumerate}
\item $\cO^\times$ can exceed $F^*$: $x+\sqrt{x^2+1}$ has norm $1$. For $f=\sqrt{x^2+1}$ the set $S_f$ is empty (the branch primes have $v_P(f)=1>-2$), and indeed $\int\sqrt{x^2+1}\,dx=\tfrac x2\sqrt{x^2+1}+\tfrac12\log\bigl(x+\sqrt{x^2+1}\bigr)$.
\item Primes of $\cO$ over a prime $p$ of $R$ may need to be separated: $\int\frac{dx}{(x-1)\sqrt x}=\log\frac{\sqrt x-1}{\sqrt x+1}$, whose logand has divisor $P_+-P_-$ for the two primes $P_\pm$ over $x-1$. The irreducible factor $x-1$ itself, with divisor $P_++P_-$, does not suffice.
\item The element $y$ is never needed as a logand, since $\log y=\frac1m\log q$.
\end{enumerate}

\begin{remark}[on Risch--Norman condition (iii)]
$(\Fbar\otimes\cO)_S^\times/\Fbar^*$ is finitely generated for every finite $S$, since $\cO$ is a finitely generated domain over a field \cite{Samuel66}. So a finite generating set of logands exists, but it depends on $S_f$, that is on the support of $\den(f)$, and not only on $L$. Davenport--Trager's formulation of condition~(iii) --- a finite set depending only on the field --- cannot be retained for radical extensions; the correct formulation is that the set depends on the support of the denominator of the integrand.
\end{remark}

\section{The case $n=1$: residues, $S$-units and the Jacobian}\label{sec:curve}

In this section $K=\Fbar(x)$, $Dx=1$, so $R=\Fbar[x]$, $\den_D(K)=1$, and every irreducible is normal; thus every height-one prime of $\cO$ is tame, $S_f=\{P: v_P(f)\le-\delta_P\}$, and there is no special part $s$ in Corollary~\ref{cor:shape}. Let $X$ be the smooth projective curve over $\Fbar$ with function field $L$, and $S_\infty$ its set of places over $x=\infty$; by the argument of Proposition~\ref{prop:basis} applied at $1/x$, $|S_\infty|=\gcd(m,\deg q)$, each with ramification index $m/\gcd(m,\deg q)$. The affine places of $X$ are the height-one primes of $\cO$. By Riemann--Hurwitz,
\[
2g(X)-2=-2m+\sum_j\deg Q_j\,\bigl(m-\gcd(j,m)\bigr)+m-\gcd(m,\deg q).
\]

\subsection{Residues}
For a place $P$ with uniformiser $\pi$, $dx=\frac{d\pi}{D\pi}$, so
\[
\Res_P(f\,dx)=\text{constant term of }\frac{f}{D\pi/\pi}.
\]
With $f=\rho\lambda_P\pi^{-\delta_P}+\dots$ and $D\pi/\pi=\lambda_P\pi^{-\delta_P}+\dots$ this constant term is $\rho$; so the ``coefficient of $\pi^{-\delta_P}$ relative to $D\pi/\pi$'' of Theorem~\ref{thm:sunits} is the ordinary residue of the differential $f\,dx$, and since $\Res_P(dv)=0$ and $\Res_P(du/u)=v_P(u)$ we recover the classical statement
\begin{equation}\label{eq:residues}
\Res_P(f\,dx)=\sum_ic_i\,v_P(u_i)=\frac1N\sum_k\gamma_kv_P(U_k)\qquad\text{for every affine place }P .
\end{equation}
For $n>1$ there is no global residue theorem, and the leading-coefficient formulation of Theorem~\ref{thm:sunits} is the substitute.

\subsection{$S$-units and the Jacobian}
For a finite set $S$ of affine places put $T=S\cup S_\infty$ and let $\Div^0_T\cong\Z^{|T|-1}$ be the degree-zero divisors supported on $T$. The map $u\mapsto\dv(u)$ gives an exact sequence
\begin{equation}\label{eq:exact}
1\to\Fbar^*\to\cO_S^\times\to\Div^0_T\to\Jac(X)(\Fbar),
\end{equation}
so
\[
\rk\bigl(\cO_S^\times/\Fbar^*\bigr)=|T|-1-\rk J_T,\qquad J_T:=\langle[P-P']:P,P'\in T\rangle\subseteq\Jac(X)(\Fbar).
\]

\begin{theorem}\label{thm:curve}
Let $f\in L^*$ with an elementary integral, let $S=S_f$, and for each distinct nonzero value $c$ among the residues $\Res_P(f\,dx)$, $P\in S$, let
\[
\mathcal D_c:=\sum_{P\in S,\ \Res_P(f\,dx)=c}P\ \in\Div(X),
\]
completed to degree $0$ by a divisor supported on $S_\infty$ in the unique way compatible with \eqref{eq:residues}. Then $f$ has an elementary integral over $L$ if and only if
\begin{enumerate}
\item[(i)] each class $[\mathcal D_c]$ is torsion in $\Jac(X)$, say $N_c[\mathcal D_c]=0$ and $\dv(u_c)=N_c\mathcal D_c$ with $u_c\in L^*$; and
\item[(ii)] the residue-free remainder $f-\sum_c\frac{c}{N_c}\,\frac{Du_c}{u_c}$ is a derivative in $L$,
\end{enumerate}
in which case $\int f=v+\sum_c\frac{c}{N_c}\log u_c$ with $Dv$ the remainder in (ii), and $v$ has the shape of Corollary~\ref{cor:shape}. Condition (ii) is equivalent to the solvability of the linear system of Corollaries~\ref{cor:shape} and~\ref{cor:degbound} for $v$, and it is \emph{not} implied by (i): the differential $f\,dx-\sum_c\frac{c}{N_c}\frac{du_c}{u_c}$ has zero residues everywhere but may represent a nonzero class in $H^1_{\mathrm{dR}}(X)$.
\end{theorem}
\begin{proof}
If $f$ has an elementary integral, the strong Liouville theorem together with \eqref{eq:residues} and \eqref{eq:exact} gives (i): the residues determine the divisor of each $U_k$ at every place (on $S$ directly, at $S_\infty$ by the completion clause), so $\dv(U_k)$ is a $\Z$-combination of the $N_c\mathcal D_c$ and each $\mathcal D_c$ has a principal multiple. The rearranged logarithmic part then agrees with $\sum_c\frac{c}{N_c}Du_c/u_c$ up to logarithmic derivatives of constants, so the remainder is $Dv$, giving (ii). The converse is immediate. This sharpens the usual compression of Risch's and Trager's criterion \cite{Risch70,Trager84}, in which (ii) is implicit in the reduction preceding the residue step; stated for the original integrand, (i) alone is insufficient (Remark~\ref{rem:torsionnotenough}).
\end{proof}

\begin{remark}\label{rem:torsionnotenough}
On $y^2=x^3+1$ take $f=1/\bigl((x-2)y\bigr)$: the residues are $\pm\frac13$ at $(2,\pm3)$, and $\mathcal D_{1/3}=(2,3)-\infty$ is torsion of order $6$, since $(2,3)$ generates $E(\Q)\cong\Z/6\Z$. Yet, with $a_\pm=(x+1)y\pm\bigl(\tfrac{x^3}6+2x^2-x+\tfrac53\bigr)$, $\dv(a_\pm)=6\,(2,\pm3)-6\infty$, one has
\[
\int\frac{dx}{(x-2)\sqrt{x^3+1}}=\frac1{18}\log\frac{a_-}{a_+}-\frac13\int\frac{dx}{\sqrt{x^3+1}},
\]
so condition (ii) fails and the integral is not elementary; the linear system of Theorem~\ref{thm:complete} correctly has no solution. The obstruction --- a holomorphic discrepancy between the chosen logarithmic derivative and any differential with the given residues --- is studied systematically in the companion paper \cite{PartII}, where such integrals are expressed through abelian integrals of $X$.
\end{remark}

The parallel method may therefore be described as Trager's algorithm with Hermite reduction replaced by undetermined coefficients for $b\in\cO$, and with the logand step performed either exactly (residues, torsion test, Riemann--Roch) or heuristically (a precomputed list of $S$-units). Two special cases deserve mention.

\begin{proposition}[genus 0]\label{prop:genus0}
If $g(X)=0$, every $\mathcal D_c$ is principal, so $N_c=1$, and $u_c$ is obtained by linear algebra: $u_c$ spans the one-dimensional Riemann--Roch space $\mathcal L(-\mathcal D_c)$, computed as $\{\sum_i\mu_iw_i/\prod_{p\in S}p^{k_p}: v_P(\cdot)\ge -(\mathcal D_c)_P\ \text{for }P\in T\}$ with the valuation conditions expressed through Puiseux expansions at $T$ \cite{Hess02}. In this case the logarithmic part of the integral is completely determined by the residues, and Theorem~\ref{thm:curve} together with the (heuristic) degree bound of Section~\ref{sec:algorithm} is the entire algorithm. For $m=2$, $g=0$ iff $\deg q\le2$, and the Euler substitution $z=x+y$ makes $L=\Fbar(z)$; the radical then disappears altogether.
\end{proposition}

\begin{proposition}[$m=2$: Pell equation and continued fractions]\label{prop:pell}
Let $m=2$. An element $u=a+by\in\cO_S$, $a,b\in\Fbar[x]$, is an $S$-unit iff
\[
a^2-q\,b^2=\kappa\prod_{p\in S}p^{k_p},\qquad \kappa\in\Fbar^*,\ k_p\in\Z_{\ge0},
\]
and $\dv(u)$ is read off from the $k_p$ and the valuations of $a\pm by$. In particular $\cO^\times\ne\Fbar^*$ iff $\deg q$ is even and the polynomial Pell equation $a^2-qb^2=1$ has a nontrivial solution, which by Abel's theorem \cite{Abel1826,vdPoortenTran00} holds iff the continued fraction expansion of $\sqrt q$ in $\Fbar((1/x))$ is periodic; the fundamental unit is then $a+by$ with $a/b$ the convergent at the end of the first period.
\end{proposition}
\begin{proof}
$N_{L/K}(a+by)=a^2-qb^2$, and $u$ is an $S$-unit iff its norm is an $S$-unit of $\Fbar[x]$ (both $u$ and its conjugate lie in $\cO$, and $u\bar u\in\Fbar[x]$ is an $S$-unit iff $u$ is). If $\deg q$ is odd, $|S_\infty|=1$ and $\Div^0_{S_\infty}=0$, so $\cO^\times=\Fbar^*$ by \eqref{eq:exact}. If $\deg q$ is even, $S_\infty=\{\infty_+,\infty_-\}$ and a unit has divisor $k(\infty_+-\infty_-)$; the correspondence with periodic continued fractions is Abel's theorem.
\end{proof}

For $S\ne\emptyset$ and $g\ge1$ the torsion order $N_c$ must be bounded; when the divisor $\mathcal D_c$ is defined over a number field, Trager's method \cite{Trager84,Bronstein90} bounds it by reducing $X$ modulo two primes of good reduction with distinct residue characteristics, into whose Jacobians the prime-to-$p$ torsion injects.

\section{Degree bounds}\label{sec:degree}

Risch--Norman condition (ii) --- a bound on the degree of the numerator $b$ of the rational part --- is the one part of the parallel method that has remained heuristic even for logarithmic towers \cite{DavenportTrager85,Bronstein05}; see \cite{DuRaab25} for the current state of the art, which obtains rigorous bounds from complete reduction systems in specific fields. In this section we prove exact bounds for $n=1$, a transfer theorem for the top variable of a tower, and we exhibit the obstruction for lower variables.

\subsection{The case $n=1$}\label{sec:degree1}

Let $K=\Fbar(x)$, $L=K(y)$, $X$ as in Section~\ref{sec:curve}, and let $\infty_1,\dots,\infty_g$ be the places over $x=\infty$, $g=\gcd(m,N)$, $N=\deg q$, each with $e_\infty:=m/g$. Multiplying $y$ by a constant we may assume $q$ monic. The affine analysis of Section~\ref{sec:valuations} has a uniform reformulation on the curve.

\begin{lemma}\label{lem:completion}
Let $P$ be any place of $X$ with uniformiser $\pi$. In the completion $\widehat L_P=\Fbar((\pi))$ one has $D=(D\pi)\,\frac{d}{d\pi}$. Consequently, for $g\in L^*$,
\[
v_P(Dg)=v_P(g)-1+v_P(D\pi)\quad\text{if }v_P(g)\ne0,\qquad v_P(Dg)\ge v_P(D\pi)\quad\text{if }v_P(g)=0 .
\]
For $P$ affine and tame, $v_P(D\pi)=1-\delta_P$ (Lemma~\ref{lem:poleorder}); for $P=\infty_l$, $v_P(D\pi)=1+e_\infty$.
\end{lemma}
\begin{proof}
$D=d/dx$ is continuous on $\widehat K=\Fbar((x-\xi))$ or $\Fbar((1/x))$, and its unique extension to the finite separable extension $\widehat L_P$ is continuous. A continuous derivation of $\Fbar((\pi))$ vanishing on $\Fbar$ is determined by its value on $\pi$, whence $D=(D\pi)\,d/d\pi$, and the valuation statements follow. At $\infty_l$ choose integers $\alpha,\beta$ with $\alpha N/g+\beta e_\infty=-1$ and $\pi=y^\alpha x^\beta$; then $v(\pi)=-\alpha N/g-\beta e_\infty=1$, and
\[
\frac{D\pi}{\pi}=\frac{\alpha}{m}\frac{q'}{q}+\frac{\beta}{x}=\Bigl(\frac{\alpha N}{m}+\beta\Bigr)\frac1x+O(x^{-2})=-\frac{1}{e_\infty x}+O(x^{-2}),
\]
which has valuation $e_\infty$ since $v_{\infty_l}(1/x)=e_\infty$.
\end{proof}

\begin{theorem}[pole order at infinity]\label{thm:infinity}
Let $f\in L^*$ have an elementary integral, with $v$ as in \eqref{eq:liouville}. Then for each $l$,
\[
v_{\infty_l}(v)=v_{\infty_l}(f)-e_\infty\ \text{ if } v_{\infty_l}(f)<e_\infty,\qquad v_{\infty_l}(v)\ge0\ \text{ if } v_{\infty_l}(f)\ge e_\infty .
\]
\end{theorem}
\begin{proof}
Write $v_l=v_{\infty_l}$, $e=e_\infty$. By Lemma~\ref{lem:completion}, $v_l(Du/u)=v_l\bigl((D\pi/\pi)\cdot\pi u'/u\bigr)\ge e$ for every $u\in L^*$, so $v_l(\Sigma)\ge e$ for the logarithmic part $\Sigma$. If $v_l(v)=-k<0$ then $v_l(Dv)=-k+e<e\le v_l(\Sigma)$, so $v_l(f)=-k+e$. If $v_l(v)\ge0$ then $v_l(Dv)\ge e$ by the lemma, so $v_l(f)\ge e$.
\end{proof}

To convert this into bounds on the coefficients $b_i$ we need that the terms $b_iw_i$ cannot cancel simultaneously at all places at infinity. Put
\[
\nu_i:=v_{\infty_l}(w_i)=-\frac{iN}{g}+e_\infty\sum_j\Bigl\lfloor\frac{ij}{m}\Bigr\rfloor\deg Q_j ,
\]
which is independent of $l$.

\begin{lemma}[no simultaneous cancellation]\label{lem:vandermonde}
For $a=\sum_ib_iw_i\in\cO$, $\min_lv_{\infty_l}(a)=\min_i\bigl(-e_\infty\deg b_i+\nu_i\bigr)$.
\end{lemma}
\begin{proof}
Since $\nu_i\equiv-iN/g\pmod{e_\infty}$ and $\gcd(N/g,e_\infty)=1$, terms with $i\not\equiv i'\pmod{e_\infty}$ have valuations in different residue classes mod $e_\infty$ and cannot cancel. Within a class $i=i_0+ke_\infty$, $0\le k<g$, one has $w_{i_0+ke_\infty}=w_{i_0}(y^{e_\infty})^k\cdot r_k$ with $r_k\in\Fbar(x)$, and $y^{e_\infty}=x^{N/g}z$ where $z^g=q/x^N\to1$ at infinity, so the leading coefficient of $z$ at $\infty_l$ is $\zeta^l$ for a primitive $g$-th root of unity $\zeta$. Hence the leading coefficient of $\sum_kb_{i_0+ke_\infty}w_{i_0+ke_\infty}$ at $\infty_l$, at the minimal valuation $\mu$ of its terms, is $\sum_{k\in I}\beta_k\zeta^{lk}$ with $I$ the set of indices attaining $\mu$ and $\beta_k\ne0$ independent of $l$. If this vanished for all $l=0,\dots,g-1$, the Vandermonde matrix $(\zeta^{lk})$ would force all $\beta_k=0$.
\end{proof}

\begin{corollary}[exact degree bounds, $n=1$]\label{cor:degbound}
Let $f=\sum_ia_iw_i/d$ as in \eqref{eq:rep}, $d=\prod_jd_j^{\,j}$, $\Delta:=\deg\prod_jd_j^{\,j-1}$, and
\[
\mu(f):=\min_i\bigl(-e_\infty\deg a_i+\nu_i\bigr)+e_\infty\deg d=\min_lv_{\infty_l}(f),\qquad k:=\max\bigl(0,\ e_\infty-\mu(f)\bigr).
\]
If $f$ has an elementary integral, then $v=\sum_ib_iw_i/\prod_jd_j^{\,j-1}$ with
\[
\deg b_i\ \le\ \Bigl\lfloor\frac{k+e_\infty\Delta+\nu_i}{e_\infty}\Bigr\rfloor\qquad(0\le i<m).
\]
\end{corollary}
\begin{proof}
By Theorem~\ref{thm:infinity}, $v_{\infty_l}(v)\ge-k$ for all $l$, so $a:=v\prod_jd_j^{\,j-1}\in\cO$ (Corollary~\ref{cor:shape}) satisfies $v_{\infty_l}(a)\ge-k-e_\infty\Delta$ for all $l$; Lemma~\ref{lem:vandermonde} applied to $a$ and to $f$ gives the bound and the formula for $\mu(f)$.
\end{proof}

For the three examples of Section~\ref{sec:examples} the corollary gives $(\deg b_0,\deg b_1)\le(2,1)$, $(3,1)$ and $(\deg b_0,\deg b_1,\deg b_2)\le(1,0,0)$ respectively, against the heuristic bounds $9$, $7$ and $4$ on all coordinates; re-running the prototype with these bounds reproduces the integrals with $9$, $9$ and $7$ unknowns instead of $24$, $19$ and $18$. Since Corollary~\ref{cor:shape}, Corollary~\ref{cor:degbound} and Theorem~\ref{thm:curve} now cover all three Risch--Norman conditions, we obtain:

\begin{theorem}[completeness for $n=1$]\label{thm:complete}
For $L=\Fbar(x,y)$, $y^m=q$, the ansatz of Corollaries~\ref{cor:shape} and~\ref{cor:degbound}, with logands from Theorem~\ref{thm:curve}, contains every elementary integral: Algorithm~2 with these bounds and the exact tier of \textsc{HiddenUnits} returns an elementary integral of $f$ whenever one exists, provided the torsion orders $N_c$ are bounded (which is unconditional in genus $0$, and follows from Trager's bound when the residue divisors are defined over a number field). Failure of the linear system therefore proves that $f$ has no elementary integral over $L$.
\end{theorem}

This is, in substance, Trager's algorithm \cite{Trager84} with Hermite reduction and the basis normal at infinity replaced by undetermined coefficients; the point here is that no reduction step is needed once the valuation bounds are available.

\subsection{Towers: the top variable}\label{sec:degreetop}

Now let $n\ge1$, $E=F(t_1,\dots,t_{n-1})$, $t=t_n$, $K=E(t)$, and let $\infty_t$ be the place of $K$ defined by $v_{\infty_t}(g)=-\deg_tg$ (degree of a rational function in $t$), with residue field $E$ and uniformiser $\tau=1/t$. Let $q\in R$ have $\deg_tq=N$ and leading coefficient $q_N$.

\begin{theorem}[transfer of the top-variable bound]\label{thm:top}
Assume $F$ is algebraically closed, $q_N\in F$, and $t$ is either primitive ($Dt=\eta\in E$) or hyperexponential ($Dt=\eta t$, $\eta\in E$) over $E$. Let $\infty'$ be a place of $L$ over $\infty_t$, with ramification index $e$. Then for all $g\in L^*$ with $v_{\infty'}(g)\ne0$:
\begin{enumerate}
\item[(a)] if $t$ is primitive, $v_{\infty'}(g)\le v_{\infty'}(Dg)\le v_{\infty'}(g)+e$;
\item[(b)] if $t$ is hyperexponential, $v_{\infty'}(Dg)=v_{\infty'}(g)$.
\end{enumerate}
Consequently, if $f\in L$ has an elementary integral, then with $\deg_t:=-v_{\infty'}/e$ (a rational number),
\[
\deg_tv\le\max(1,\deg_tf+1)\ \text{ in case (a)},\qquad \deg_tv\le\max(0,\deg_tf)\ \text{ in case (b)},
\]
and Lemma~\ref{lem:vandermonde} holds verbatim at the places over $\infty_t$, so that these translate into per-coordinate bounds $\deg_tb_i\le\lfloor(k+e\deg_t\den(v)+\nu_i)/e\rfloor$ with $k=\max(0,e-\min_lv_{\infty_l}(f))$ in case (a) and $k=\max(0,-\min_lv_{\infty_l}(f))$ in case (b).
\end{theorem}

\begin{proof}
Scaling $y$ by a constant we may take $q_N=1$, so $q=t^N\tilde q$ with $\tilde q\in1+\tau E[[\tau]]$. On $\widehat K=E((\tau))$ the derivation is $D=D_E+(D\tau)\,\partial_\tau$ acting coefficientwise, with $D\tau=-\eta\tau^2$ (primitive) or $-\eta\tau$ (hyperexponential); it is continuous and preserves $E[[\tau]]$. Put $g_0=\gcd(m,N)$, $e=m/g_0$, $z=y^e\tau^{N/g_0}$; then $z^g=\tilde q$ and by Hensel's lemma $\tilde q^{1/g_0}\in1+\tau E[[\tau]]$, so the places over $\infty_t$ are $\infty_l$, $0\le l<g_0$, with $z\mapsto\zeta^l\tilde q^{1/g_0}$, each totally ramified of index $e$ with residue field $E$, and $\widehat L_{\infty_l}=E((\pi))$ with $\pi=y^\alpha\tau^\beta$, $-\alpha N/g_0+\beta e=1$. The coefficient field $E$ is $D$-stable, $D$ extends continuously, and hence $D=D_E+(D\pi)\,\partial_\pi$ on $E((\pi))$. Now
\[
\frac{D\pi}{\pi}=\frac\alpha m\frac{Dq}{q}+\beta\frac{D\tau}{\tau},\qquad \frac{Dq}{q}=N\frac{Dt}{t}+\frac{D\tilde q}{\tilde q},\qquad \frac{D\tilde q}{\tilde q}\in\tau E[[\tau]],
\]
and $\frac{\alpha N}{m}-\beta=-\frac1e$. Also $\tau/\pi^e=z^{-\alpha}$ has residue $\zeta^{-l\alpha}\in F^*$.

\emph{Primitive case.} $D\pi/\pi=-\frac\eta e\tau+O(\tau^2)$ has valuation $e$ and leading coefficient $\lambda:=-\frac\eta e\zeta^{-l\alpha}\in E^*$ at $\pi^e$. For $g=\sum_{k\ge k_0}c_k\pi^k$ with $c_{k_0}\neq0$,
\[
Dg=\sum_k(D_Ec_k)\pi^k+\sum_kkc_k\pi^k\frac{D\pi}{\pi},
\]
so $v(Dg)\ge k_0$, and the coefficients of $\pi^{k_0},\dots,\pi^{k_0+e}$ in $Dg$ are $D_Ec_{k_0},\dots,D_Ec_{k_0+e-1}$ and $D_Ec_{k_0+e}+k_0c_{k_0}\lambda$. If they all vanished, then $c_{k_0}\in\Const(E)=F$ and $D_Ec_{k_0+e}=\gamma\eta$ with $\gamma=k_0c_{k_0}\zeta^{-l\alpha}/e\in F^*$ (using $k_0\ne0$), i.e.\ $\eta\in D_EE$, contradicting the fact that $t$ is a primitive monomial over $E$ \cite[Thm.~5.1.1]{Bronstein05}. Hence $v(Dg)\le k_0+e$.

\emph{Hyperexponential case.} $D\pi/\pi=-\frac\eta e+O(\pi)$, so the coefficient of $\pi^{k_0}$ in $Dg$ is $D_Ec_{k_0}-\frac{k_0}{e}\eta c_{k_0}$. If it vanished, $c:=c_{k_0}\in E^*$ would satisfy $D(c^e t^{-k_0})=c^et^{-k_0}(eD_Ec/c-k_0\eta)=0$, so $c^e=\gamma t^{k_0}$ with $\gamma\in F$, impossible for $k_0\ne0$ since $t$ is transcendental over $E$. Hence $v(Dg)=k_0$.

The consequences follow as in Theorem~\ref{thm:infinity}: $v(\Sigma)\ge0$ at $\infty'$ because $D\mathcal O_{\infty'}\subseteq\mathcal O_{\infty'}$ and $v(D\pi/\pi)\ge0$; if $v(v)=-k<0$ then in case (a) either $v(Dv)<0$, so $v(f)=v(Dv)\le-k+e$, or $v(Dv)\ge0$, so $k\le e$; in case (b) $v(f)=v(Dv)=-k$. The Vandermonde argument of Lemma~\ref{lem:vandermonde} only used that the leading coefficients of $z$ at the $g_0$ places are the $g_0$-th roots of unity, which is the case here because $q_N\in F$.
\end{proof}

\begin{remark}
(a) The hypothesis $q_N\in F$ cannot simply be dropped: if $q_N\notin F$ the residue of $D\pi/\pi$ at $\infty'$ is $\frac{\alpha}{m}Dq_N/q_N\ne0$ even in the primitive case, so $y$ behaves at $t=\infty$ like an exponential of $\frac1m\log q_N$, and the cancellation analysis changes.

(b) Theorem~\ref{thm:top} together with Corollary~\ref{cor:shape} shows that the radical never worsens the degree bound in the top variable: whatever bound holds for $K$ holds for $L$ with the same shift ($+1$ for primitive, $0$ for hyperexponential), applied coordinatewise through $\nu_i$. This is Bronstein's heuristic \cite[(10.4)]{Bronstein05} made rigorous in the top variable.
\end{remark}

\subsection{Lower variables: the obstruction}\label{sec:degreelower}

The proof of Theorem~\ref{thm:top} used the hypothesis on $t$ in two places: $D$ preserves the valuation ring at $\infty_t$ with residue derivation $D_E$, and $\eta\notin D_EE$. For a lower variable $t_i$ of a logarithmic tower the first property still holds, but the residue derivation $\bar D_i$ on $\kappa_i=F(t_1,\dots,\widehat{t_i},\dots,t_n)$ sends each later generator $t_j$ to the leading coefficient of $Dt_j$ with respect to $t_i$, and $\overline{Dt_i}$ may well lie in $\bar D_i\kappa_i$. When it does, cancellation of two orders is possible and the ``$+1$'' bound fails. We record two instances.

\begin{example}\label{ex:nestedlog}
Let $K=\Q(x,t_1,t_2)$ with $Dt_1=1/x$ and $Dt_2=(1+t_1)/(xt_1)$, i.e.\ $t_1=\log x$, $t_2=\log(x\log x)$; both are primitive monomials. For $g=t_1^k-kt_2t_1^{k-1}$,
\[
Dg=-\frac{k}{x}\,t_1^{k-2}\bigl(1+(k-1)t_2\bigr),
\]
so $f:=Dg$ has $\deg_{t_1}f=k-2$ while $\deg_{t_1}\int f=k$: the degree drops by two. Here $\overline{Dt_1}=1/x=\bar D_1t_2$. The heuristic bound \cite[(10.4)]{Bronstein05} gives $\deg_{t_1}b\le\max(\deg_{t_1}a,\deg_{t_1}d)+1=k-1$ (with $d=x$), so \textsc{ParallelIntegrate} fails on
\[
\int-\frac{2}{x}\bigl(1+\log(x\log x)\bigr)\,dx=\log^2x-2\log x\,\log(x\log x).
\]
\end{example}

\begin{example}\label{ex:irredlog}
Davenport--Trager normalise the generators to be \emph{irreducibly logarithmic}, which rewrites $\log(x\log x)$ as $\log x+\log\log x$ and removes the previous example. The phenomenon persists, however: let $\theta_1=\log x$ and $\theta_2=\log(x\theta_1+1)$, whose argument is irreducible in $\Q[x,\theta_1]$. With $g=\theta_1^2-2\theta_1\theta_2$,
\[
f=Dg=-\frac{2\bigl(x\theta_1\theta_2+x\theta_1-\theta_1+\theta_2\bigr)}{x\,(x\theta_1+1)} ,
\]
whose degree in $\theta_1$ in the sense of \cite[Def.~1]{DavenportTrager85} (numerator degree minus denominator degree) is $0$, while $\deg_{\theta_1}g=2$; for general $k$ one gets $k-2$ against $k$. So Risch--Norman condition (ii) of \cite{DavenportTrager85} is violated by an irreducibly logarithmic tower. Bronstein's bound (10.4), which uses $\max(\deg a,\deg d)$ instead, gives $k$ here and is not violated; whether it holds for all nested logarithmic towers remains open.
\end{example}

Both examples are verified by direct differentiation. They show that a proof of condition (ii) for lower variables must control the residue derivation $\bar D_i$, whose constant field may be strictly larger than $F$; the reduction-system approach of \cite{DuRaab25} is the natural framework for this. The radical extension itself is not the obstacle.

\section{The algorithm}\label{sec:algorithm}

We keep Bronstein's conventions \cite[\S10.3]{Bronstein05}. Steps that take place in $R$ (\textsc{SplitFactor}, \textsc{SquareFree}, \textsc{IrreducibleFactors}, the choice of the special set $\mathcal S$ and of the special denominator $v_s$) are unchanged. For $n=1$ the degree bounds of Corollary~\ref{cor:degbound} are used, and they are exact; for towers the bound in the top variable is that of Theorem~\ref{thm:top}, while in lower variables we keep Bronstein's heuristic \cite[(10.4)]{Bronstein05} applied to each coordinate $a_i$, with the caveat of Section~\ref{sec:degreelower}. The algorithm is exhaustive for the denominator (Corollary~\ref{cor:shape}) and, when $n=1$, for the numerator degrees; for the logands it is exhaustive if the set $\mathcal U$ generates $\cO_{S_f}^\times$ (Theorem~\ref{thm:sunits}), which for $n=1$ can be guaranteed by Theorem~\ref{thm:curve} at the cost of the torsion test.

\begin{algorithm}[H]
\caption{\textsc{RadicalSetup}$(K,D,q,m)$ --- once per field}
\begin{algorithmic}[1]
\State reduce $q$ to be $m$-th-power-free, adjusting $y$ \Comment{(N1)}
\State ensure $\gcd(m,\{j:Q_j\ne1\})=1$; otherwise replace $(q,m)$ by the generating radical \Comment{(N2)}
\State $(Q_1,\dots,Q_{m-1})\gets\textsc{SquareFree}(q)$;\quad $Q\gets\prod_jQ_j$
\State $E_i\gets\prod_jQ_j^{\lfloor ij/m\rfloor}$ for $0\le i<m$ \Comment{$w_i=y^i/E_i$}
\State $\mathrm{Mul}[i][k]\gets\bigl((i+k)\bmod m,\ \prod_jQ_j^{\lfloor(i+k)j/m\rfloor-\lfloor ij/m\rfloor-\lfloor kj/m\rfloor}\bigr)$ \Comment{\eqref{eq:multtable}}
\State $\Lambda_i\gets\sum_j\{ij/m\}\,DQ_j/Q_j$ \Comment{$Dw_i=\Lambda_iw_i$, \eqref{eq:Dw}}
\State $H\gets\operatorname{lcm}(\den_D(K),Q)$;\quad $(h_n,h_s)\gets\textsc{SplitFactor}(H,D)$
\State $\Delta\gets1+\max\bigl(0,\ \deg H-\min(\min_i\deg(H\,Dt_i),\ \min_i\deg(H\Lambda_i))\bigr)$ \Comment{total degree}
\State $\mathcal S\gets$ specials of $K$ as in \cite[\S10.3]{Bronstein05};\quad $\mathcal S\gets\mathcal S\cup\textsc{IrreducibleFactors}(h_s)$
\State $\mathcal U\gets\textsc{HiddenUnits}(q,m,\textsc{IrreducibleFactors}(Q))$ \Comment{generators of $\cO^\times_{S_{\mathrm{branch}}}$, Section~\ref{sec:curve}}
\State \Return $(E,\mathrm{Mul},\Lambda,H,\Delta,\mathcal S,\mathcal U)$
\end{algorithmic}
\end{algorithm}

\begin{algorithm}[H]
\caption{\textsc{ParallelIntegrateRadical}$(f,D)$}
\begin{algorithmic}[1]
\State write $f=\bigl(\sum_ia_iw_i\bigr)/d$ with $a_i,d\in R$, $\gcd(a_0,\dots,a_{m-1},d)=1$ \Comment{\eqref{eq:rep}}
\State $(d_n,d_s)\gets\textsc{SplitFactor}(d,D)$;\quad $(d_1,\dots,d_e)\gets\textsc{SquareFree}(d_n)$
\State $\{p_1,\dots,p_t\}\gets\textsc{IrreducibleFactors}(d_1\cdots d_e)$ \Comment{over $F$ or $\Fbar$, cf.\ \cite[\S10.3]{Bronstein05}}
\State $v_s\gets d_s\prod_{p\in\mathcal S,\,p\nmid d_s}p$
\If{each $t_k$ is a primitive, hyperexponential or hypertangent monomial}
   \For{$k=1,\dots,n$}
      \State $B_k\gets\max(\max_i\deg_{t_k}a_i,\ \deg_{t_k}d)$;\quad \textbf{if} $\deg_{t_k}(Dt_k)=0$ \textbf{then} $B_k\gets B_k+1$
   \EndFor
   \State $b_i\gets\sum_{j_1\le B_1,\dots,j_n\le B_n}\mu_{i,\mathbf j}\,t^{\mathbf j}$ for $0\le i<m$ \Comment{unknown $\mu_{i,\mathbf j}\in F$}
\Else
   \State $\delta\gets\max_i\deg a_i+\Delta$;\quad $b_i\gets\sum_{|\mathbf j|\le\delta}\mu_{i,\mathbf j}\,t^{\mathbf j}$
\EndIf
\State $\mathcal U_f\gets\mathcal U\ \cup$ (optional) $S_f$-units from residues, Theorem~\ref{thm:curve} \Comment{exact tier, $n=1$}
\State $v\gets\bigl(\sum_ib_iw_i\bigr)\big/\bigl(v_s\prod_jd_j^{\,j-1}\bigr)$
\State $g\gets v+\sum_{s\in\mathcal S}\alpha_s\log s+\sum_i\beta_i\log p_i+\sum_{u\in\mathcal U_f}\gamma_u\log u$
\State compute $Dg$ using $D(b_iw_i)=(Db_i+b_i\Lambda_i)w_i$, $D\log u=Du/u$, and $\mathrm{Mul}$ to return to the basis $\{t^{\mathbf j}w_i\}$
\State clear denominators in $f=Dg$ and equate coefficients of $t^{\mathbf j}w_i$: a linear system over $F$ in $\mu,\alpha,\beta,\gamma$
\If{the system has no solution} \Return ``failed''
\Else\ \Return $g$ with a solution substituted (free parameters set to $0$)
\EndIf
\end{algorithmic}
\end{algorithm}

\paragraph{\textsc{HiddenUnits}.} Three strategies, of increasing cost:
\begin{enumerate}
\item $n=1$, $m=2$: expand $\sqrt q$ as a continued fraction in $F((1/x))$ up to a period bound; a period yields the fundamental unit (Proposition~\ref{prop:pell}). $S$-units supported on branch primes are found similarly from $a^2-qb^2=\kappa\prod p^{k_p}$ with bounded $k_p$.
\item general $n,m$: a bounded-degree ansatz $u=\sum\mu_iw_i$ with the polynomial condition that $N_{L/K}(u)$ be a product of allowed primes, solved as in Bronstein's \textsc{FindSpecials}. Depends only on the field and can be cached.
\item $n=1$, exact: residues, torsion bound and Riemann--Roch spaces (Theorem~\ref{thm:curve}).
\end{enumerate}

\begin{remark}
As in \cite[Ex.~10.3.1]{Bronstein05}, the guess $v_s$ for the special denominator can fail when a branch prime is also special or divides $\den_D(K)$; nothing in Section~\ref{sec:structure} bounds those exponents.
\end{remark}

\section{Examples}\label{sec:examples}

The examples below were run on a prototype of Algorithms~1--2 for $K=\Q(x)$ (about 350 lines of SymPy), with the exact tier of \textsc{HiddenUnits}: residues at the places over the factors of $\den(f)$, then, for each residue value $c$, a search for $a\in\cO$ with affine divisor $N\mathcal D_c$ ($N=1,2,\dots$) by the linear algebra of Proposition~\ref{prop:genus0}, with the norm of $a$ used to certify that no further affine zeros occur; units at infinity for $m=2$ come from the continued fraction of $\sqrt q$. Linear systems are solved exactly over the number field generated by the constants that appear. In each case we also report which candidate logand sets are \emph{insufficient}, to show that the new ingredients of Theorem~\ref{thm:sunits} are needed.

\begin{example}[genus 0, two independent residues, new constants]\label{ex:genus0}
Let $y^2=x^2+1$ and
\[
f=\frac{2x^2-5}{(x^2-2)(x^2-3)\,y}+\frac{x}{y}+\frac{1}{x^2y}
 =\frac{x^7-5x^5+3x^4+6x^3-10x^2+6}{x^2(x^2-2)(x^2-3)(x^2+1)}\,w_1 .
\]
Here $d_1=(x^2-2)(x^2-3)(x^2+1)$ and $d_2=x$, so the ansatz is $v=(b_0+b_1y)/x$ with $\deg b_i\le9$. The residues are $\pm\sqrt6/12$ at the four places over $x=\pm\sqrt2$ (where $y=\pm\sqrt3$), $\pm\sqrt3/12$ at the four places over $x=\pm\sqrt3$ (where $y=\pm2$), and $0$ at the two places over $x=0$; the branch places over $x=\pm i$ have $v_P(f)=-1>-2=-\delta_P$ and are not in $S_f$. For $c=\sqrt6/12$, $\mathcal D_c=P_{(\sqrt2,\sqrt3)}+P_{(-\sqrt2,-\sqrt3)}$ and the search returns $a=\sqrt2\,y-\sqrt3\,x$ (norm $-(x^2-2)$); the other three divisors give the conjugate elements. The linear system (24 unknowns, 28 equations) has a one-dimensional solution space, and
\[
\int f\,dx=y-\frac yx+\frac{\sqrt6}{12}\log\frac{\sqrt2\,y-\sqrt3\,x}{\sqrt2\,y+\sqrt3\,x}+\frac{\sqrt3}{12}\log\frac{2x-\sqrt3\,y}{2x+\sqrt3\,y}.
\]
Neither $\sqrt2$, $\sqrt3$ nor $\sqrt6$ occurs in the integrand. With the logands restricted to the irreducible factors of $d$ over $\Q$ (or even over $\Fbar$) the system has no solution: the places over $x=\sqrt2$ must be separated, as in Theorem~\ref{thm:sunits}(2). This is the radical analogue of \cite[Ex.~10.3.2]{Bronstein05}. Running time: 12\,s.
\end{example}

\begin{example}[genus 1, a unit at infinity and a 2-torsion divisor]\label{ex:genus1}
Let $y^2=x^4+1$, a curve of genus $1$, and
\[
f=\frac{x^2+1}{xy}+\frac{x^3}{y}+\frac{x^4-1}{x^2y}=\frac{x^5+x^4+x^3+x-1}{x^2(x^4+1)}\,w_1 .
\]
Then $d_1=x^4+1$, $d_2=x$, ansatz $v=(b_0+b_1y)/x$. The residues are $\pm1$ at $P_\pm=(0,\pm1)$ and $0$ at the four branch places, which again lie outside $S_f$. The class $[P_+-P_-]$ is nontrivial $2$-torsion in $\Jac(X)$: no element of $\cO$ has affine divisor $P_+$ or $2P_+$, and the search first succeeds at $N=4$ with $a=y-1$ (divisor $4P_+-2\infty_+-2\infty_-$) and $y+1$. Since $\deg q$ is even, the continued fraction $\sqrt{x^4+1}=[\,x^2;\overline{2x^2}\,]$ is periodic and yields the unit $x^2+y$ of norm $-1$. The solver returns
\[
\int f\,dx=\frac{y}{2}+\frac{y}{x}+\frac12\log\bigl(x^2+y\bigr)+\frac14\log\frac{y-1}{y+1},
\]
where $\frac14\log\frac{y-1}{y+1}=\log x-\frac12\log(1+y)+\mathrm{const}$ exhibits the torsion order $2$ directly. Without the unit $x^2+y$, or with $x$ as the only logand, there is no solution: an integrand with vanishing affine residues may still require a logarithm, and it is one of a unit of $\cO$. Running time: 1\,s.
\end{example}

\begin{example}[$m=3$, nontrivial integral basis, three logarithms]\label{ex:cubic}
Let $y^3=x^2$. Then $Q_2=x$, $E=(1,1,x)$, $w_2=y^2/x$ (that is, $x^{1/3}$), $\Lambda=(0,\tfrac{2}{3x},\tfrac1{3x})$, and the unique branch place $P_0$ over $x=0$ has $e_{P_0}=3$. Take
\[
f=\frac{1}{x^{2/3}(x-1)}+x^{-4/3}=\frac{w_2}{x(x-1)}+\frac{w_1}{x^2},
\]
so $d=x^2(x-1)$, $d_1=x-1$, $d_2=x$, and $v=(b_0+b_1w_1+b_2w_2)/x$. The three places over $x=1$ are $y=\omega^k$, $\omega=e^{2\pi i/3}$, with residues $\omega^k$; at $P_0$, $v_{P_0}(f)=v_{P_0}(w_1/x^2)=2-6=-4\le-3$, so $P_0\in S_f$, but its residue (computed from the trace) is $0$. The divisors $\mathcal D_{\omega^k}=P_k$ are principal, with $a_k=w_2-\omega^k$ (divisor $P_k-\infty$), and
\[
\int f\,dx=-\frac{3y}{x}+\sum_{k=0}^2\omega^k\log\bigl(w_2-\omega^k\bigr)=-3x^{-1/3}+\sum_{k=0}^2\omega^k\log\bigl(x^{1/3}-\omega^k\bigr).
\]
The constant $\omega$ is not in the integrand; with the logands $x$ and $x-1$ only, there is no solution. Running time: 2\,s.
\end{example}

\begin{example}[radical over an exponential]
$K=\Q(x,t)$, $Dt=t$, $q=t$, $m=2$: $y=e^{x/2}$. Here $Q_1=t$, but $t$ is special in $R$, so the branch prime is not tame (Remark~\ref{rem:special}(b)); $H=t$, $h_s=t$, $\mathcal S=\{t\}$, and the special denominator $v_s$ is guessed. The algorithm reduces to Bronstein's for the field $\Q(x,e^{x/2})$, as it should.
\end{example}

\section{Conclusion and open problems}\label{sec:conclusion}

We have shown that the two structural facts on which the parallel Risch method rests --- the Hermite-type shape of the denominator and the identification of possible logarithms --- extend to a simple radical extension of a differential field $F(t_1,\dots,t_n)$, with the polynomial ring $R$ still hosting all factorisation, and with the logarithms now given by $S$-units of the integral closure rather than by irreducible polynomials. For $n=1$ the $S$-unit problem is the classical torsion problem in the Jacobian, trivial in genus $0$ and governed by the polynomial Pell equation when $m=2$.

For $n=1$ the degree bound is exact and the method is complete. Open questions include: (a) a proof of a degree bound in lower variables of a tower; Section~\ref{sec:degreelower} shows that the bound must depend on the residue derivation at $t_i=\infty$ and not only on the degrees of the integrand, and that the Davenport--Trager form of condition (ii) is false as stated; (b) several radicals $y_1^{m_1}=q_1,\dots,y_r^{m_r}=q_r$, where an explicit integral basis exists when the $q_i$ have pairwise coprime discriminants but not in general; (c) $S$-unit generators for $n>1$, where the Jacobian is replaced by the Picard group of a higher-dimensional variety --- for a constant curve ($q\in F[t_1]$) this is settled in \cite{PartII} via $\Cl(X^\circ\times\A^{n-1})\cong\Cl(X^\circ)$; (d) whether the approach gives a useful preprocessor for the complete algebraic algorithms of \cite{Trager84,Bronstein90} in practice.

\section*{Acknowledgements}
I had fruitful discussions on the possibility of this result on different occasions with Albert Rich.

\end{document}